\documentclass[11pt,letterpaper]{article}

\usepackage{fullpage}

\usepackage{graphicx} 

    \usepackage{algorithm}
    \usepackage{algorithmic}

    \usepackage[square,numbers]{natbib}

    \usepackage{hyperref}
    \usepackage{graphicx}
    \usepackage{svg}
    \usepackage{float}
    \usepackage{verbatim}

    \usepackage[utf8]{inputenc} 
\usepackage[T1]{fontenc}    
\usepackage{hyperref}       
\usepackage{url}            
\usepackage{booktabs}       
\usepackage{amsfonts}       
\usepackage{nicefrac}       
\usepackage{microtype}      
\usepackage{xcolor}         
\usepackage[normalem]{ulem} 
\usepackage{amsmath,amssymb}
\usepackage{amsthm}
\usepackage[capitalize]{cleveref}

    \makeatletter
\@ifundefined{originalleft}{\let\originalleft\left}{}
\@ifundefined{originalright}{\let\originalright\right}{}
\makeatother
   \newtheorem{theorem}{Theorem}[section]
\newtheorem{lemma}[theorem]{Lemma}
\newtheorem{fact}[theorem]{Fact}
\newtheorem{claim}[theorem]{Claim}
\theoremstyle{definition}
\newtheorem{definition}[theorem]{Definition}

    \renewcommand{\left}{\mathopen{}\mathclose\bgroup\originalleft}
    \renewcommand{\right}{\aftergroup\egroup\originalright}
    \renewcommand{\epsilon}{\varepsilon}
    \newcommand{\eps}{\varepsilon}
    \renewcommand{\backslash}{\setminus}

\newcommand{\dens}{\rho}

\newcommand{\Comment}[1]{\textcolor{black!80}{\texttt{// #1}}}

\newcommand{\PO}{\textsf{PO}}

\newcommand{\EFx}{\textsf{EFx}}

\newcommand{\ComputeFEFx}{\textsc{ComputeFEFx}}
\newcommand{\ComputeFEFeps}{\textsc{ComputeFEF-eps}}

\newcommand{\FEF}{\ensuremath{\mathsf{FEF}}}
\newcommand{\FEFx}{\ensuremath{\mathsf{FEF}_x}}
\newcommand{\charity}{\ensuremath{\mathrm{charity}}} 

\title{\bfseries Existence and Computation of Fair Allocations\\under Constraints}

\author{Siddharth Barman\thanks{Department of Computer Science and Automation, Indian Institute of Science, Bangalore 560012, India. {\tt barman@iisc.ac.in}} \and Ioannis Caragiannis\thanks{Department of Computer Science, Aarhus University, {\AA}bogade 34, 8200 Aarhus N, Denmark. {\tt iannis@cs.au.dk}} \and Sudarshan Shyam\thanks{Department of Computer Science, Aarhus University, {\AA}bogade 34, 8200 Aarhus N, Denmark. {\tt  shyam@cs.au.dk}} }
\date{\empty}

\begin{document}



\maketitle

\begin{abstract}
We study fair division of divisible goods under generalized assignment constraints. Here, each good has an agent-specific value and size, and every agent has a budget constraint that limits the total size of the goods she can receive. Since it may not always be feasible to assign all goods to the agents while respecting the budget constraints, we use the construct of charity to accommodate the unassigned goods. In this constrained setting with charity, we obtain several new existential and computational results for feasible envy-freeness ($\FEF$); this fairness notion requires that agents are envy-free, considering only budget-feasible subsets.  

First, we simplify and extend known existential results for $\FEF$ allocations. Then, we show that the space of $\FEF$ allocations has a non-convex structure. Next, using a fixed-point argument, we establish a novel guarantee that $\FEF$ can always be achieved with Pareto-optimality. Furthermore, we give an alternative proof of the fact that one cannot additionally obtain truthfulness in this context: There does not exist a mechanism that is simultaneously truthful, fair, and Pareto-optimal. On the positive side, we show that truthfulness is compatible with each of $\FEF$ and Pareto-optimality, individually.
\end{abstract}

\section{Introduction}
Dividing items fairly among agents with individual preferences is a central requirement in many domains. Problems ranging from distributing land among communities to allocating computing resources among organizational groups can all be addressed using fair division frameworks. Motivated by such considerations, over the past seventy years, research efforts in the economics, computer science, and mathematics literature have resulted in significant developments in fair division \cite{moulin2004fair, handbook2016, amanatidis2022fair}. 


Classic work focused on the fair division of a single divisible and heterogeneous resource, referred to as a cake~\cite{P16}. Starting with this cake cutting setup, the fairness notion that has been predominantly studied in the literature is {envy-freeness}~\cite{F66,V74}. An allocation is said to be envy-free if no agent values the bundle of any other agent higher than her own. Another desired property when dividing resources is economic efficiency. A standard way of formalizing this desideratum is Pareto-optimality. Specifically, an allocation is said to be Pareto-optimal if we cannot increase the value of any agent, via  reallocations, without reducing that of someone else.


Notably, in the divisible-goods setting, fairness and efficiency---formalized through envy-freeness and Pareto-optimality, respectively---can be achieved together \cite{V74}. Specifically, it is known that an allocation that maximizes the Nash social welfare (the geometric mean of agents' valuations) is both envy-free and Pareto-optimal. Furthermore, one can maximize the Nash social welfare by solving the well-known convex program of Eisenberg and Gale~\cite{eisenberg1959consensus}. Hence, in the current context of divisible goods, fair and efficient allocations can, in fact, be computed in polynomial time. The interplay of fairness and efficiency is a central theme in fair division literature and has been studied in various other contexts, including in the setting of indivisible goods \cite{caragiannis2019unreasonable}. 


However, most existing fair division results (with or without efficiency considerations) assume that the agents have no restrictions on the goods they can receive. This assumption does not always capture real-world scenarios. For instance, in fair land division, a natural requirement is to ensure that each agent receives a connected plot~\cite{SNHA20}. Motivated by such considerations, recent works have considered fair division under allocation constraints \cite{suksompong2021constraints}. 


The current work contributes to this active thread of research and establishes fairness and efficiency guarantees under generalized assignment constraints. These constraints capture the allocation requirements inherent in numerous real-world domains; see multiple applications surveyed in \cite{oncan2007survey}. Here, each divisible good has an agent-specific size and value, and each agent has a budget which constrains the size of her bundle. Also, note that, in such a constraint setting, it may not always be feasible to assign all goods to the agents while respecting the budget constraints. Hence, as in recent works (see, e.g., \cite{BKSS23}), we conform to the construct of charity to denote the set of unassigned goods.

In this constrained setting with charity, our work establishes fairness guarantees considering a notion called feasible envy-freeness (FEF)~\cite{dror,wu2021budget,BKSS23}. This notion extends envy-freeness and is formulated with the idea that it is unreasonable to expect two agents, with significantly different budgets, to receive bundles of comparable value. Hence, FEF requires that agents are envy-free, considering only budget-feasible subsets. Specifically, an allocation is said to be feasibly envy-free (FEF) if and only if each agent $i$ values her bundle over every budget-feasible subset\footnote{For ease of exposition, even for divisible goods, we use the term subset to denote
fractional assignments of the goods.} within any other agent's bundle; an analogous guarantee is required against the charity.

\subsection{Our Results and Contributions}
This section details our fairness, efficiency, and truthfulness results for allocating divisible goods among agents with additive valuations and under generalized assignment constraints. Our contribution is summarized pictorially in the triangle below.

\begin{center}
\includegraphics[scale=0.6]{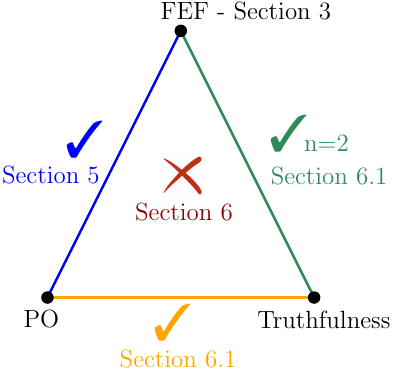}    
\end{center}

\noindent
{\bf Fairness:} Prior work \cite{BKSS23} has shown that, under generalized assignment constraints, $\FEF$ allocations of divisible goods always exist and can be found efficiently. We develop an alternate proof of the existence for $\FEF$ allocations. Our result, in fact, generalizes the previously known existential guarantee by encapsulating a more general class of constraints and valuations (Theorem \ref{theorem:generalconstraints} in Section \ref{fefgeneral}). Our proof here is non-constructive and, hence, does not directly lead to an algorithm. We also present a fair division instance with simple constraints, in which the space of $\FEF$ allocations is not convex (Theorem \ref{thm:non-convex} in Section \ref{sec:non-convex}). This comes in contrast to the unconstrained setting where envy-free allocations can be easily characterized by a set of linear inequalities. This places in context the earlier algorithm of \citet{BKSS23}, which employs a series of parameterized linear programs in the computation of an $\FEF$ allocation.
\medskip

\noindent
{\bf Fairness and Efficiency:} We establish (in Section~\ref{sec:fefpo}) that, under generalized assignment constraints, there always exists an $\FEF$ allocation that is also Pareto-optimal (among all feasible allocations). This fairness plus efficiency guarantee resolves, in the positive, an open question posed in  \cite{BKSS23}. Furthermore, our proof implies the {\rm PPAD}-membership of computing $\FEF$ and Pareto-optimal (PO) allocations. It is relevant to note that a special case of this computational problem is already known to be {\rm PPAD}-hard~\cite{vaziranippad}. 
\medskip

\noindent
{\bf Fairness, Efficiency, and Truthfulness:} Finally, we show (in Section~\ref{sec:counterexample}) that one cannot additionally achieve truthfulness in this context: There does not exist a truthful mechanism for finding $\FEF$ and $\PO$ allocations. Interestingly, our impossibility result is not confined to the constrained setting, i.e., it also holds in the standard, unconstrained setup. This instantiation to the unconstrained context is interesting in and of itself,\footnote{In particular, this negative result implies that one cannot extend the classic work of Varian \cite{V74} to include truthfulness.} and we would not be surprised if it is a folklore result; still, to the best of our knowledge, there is no reference for it in the literature.




\subsection{Additional Related Work}
Complementary to the current work, which primarily addresses divisible goods, recent works in fair division have also studied constraints in the context of indivisible goods~\cite{suksompong2021constraints}. In particular, fairness guarantees under budget constraints have been obtained for Nash social welfare \cite{budgetfeasiblenash} and maximin shares \cite{hummel2022maximin}. Also, Barman et al.~\cite{BKSS23} show that feasible envy-freeness up to any good ($\FEFx$) is achievable for indivisible goods and under generalized assignment constraints. 

\citet{gp22,Momi2017} also consider the compatibility of fairness, efficiency, and incentives: in particular, they show that imposing Pareto-optimality together with truthfulness essentially forces dictatorship, and hence rules out achieving envy-freeness as well. Our counterexample is complementary, showing that Pareto-optimality + envy-freeness already precludes (even approximate) truthfulness, yielding an “inapproximability of truthfulness” guarantee via instances where misreporting improves an agent’s value by a constant factor.

We focus on goods (i.e., items with nonnegative values). Fairness guarantees for chores (negatively valued items) under budget constraints are addressed in~\cite{choresbudgets}.




\section{Notation and Preliminaries}
\noindent
We study the problem of fair division of $m$ divisible goods among $n$ agents with additive valuations and under constraints. We use $[n]=\{1, \ldots, n\}$ and $[m]=\{1,\ldots, m\}$ to denote the set of agents and the goods, respectively. An allocation corresponds to a collection of $n$ vectors $x_1, x_2, \ldots, x_n \in [0,1]^m$. Here, $x_i = (x_{i,1}, x_{i,2}, \dots, x_{i,m})$ is the bundle assigned to agent $i \in [n]$, with $x_{i,g} \in [0,1]$ denoting the fraction of good $g \in [m]$ assigned to agent $i$. For an allocation $(x_1, \ldots, x_n)$ it must hold that $\sum_{i=1}^n x_{i,g} \leq 1$, for each good $g \in [m]$. 

We will, throughout, write $v_i: [0,1]^m \rightarrow \mathbb{R}_+$ to denote the valuations of the agent $i \in [n]$.  Unless stated otherwise, the valuations are assumed to be additive. That is, for each bundle $x_i \in [0,1]^m$, agent $i$'s valuation is defined as $v_i(x_i) = \sum_{g \in [m]} x_{i,g} \  v_{i,g}$; here, $v_{i,g} \in \mathbb{R}_+$ is the value that agent $i$ has for the entire good $g$.

In a general form of the constrained fair division problem, we have a constraint set $\mathcal{F}_i \subseteq [0,1]^m$, for each agent $i \in [n]$, and an allocation $(x_1, \ldots, x_n)$ is said to be \emph{feasible} if $x_i \in \mathcal{F}_i$, for all $i \in [n]$. Note that, in the constrained fair division setting, it may not always be possible to allocate all the goods to the agents; consider, for example, the case wherein $\mathcal{F}_i = \emptyset$ for each $i \in [n]$. To address this, we introduce the construct of \emph{charity} that receives the remaining (unassigned) fractions of the goods. Formally, for any allocation $x$ and each good $g$, we write $x_{\charity,g} = 1 - \sum_{i=1}^n x_{i,g}$. Hence, $x_{\charity}\in [0,1]^m$ denotes the unassigned bundle. Also, for bundles $y, z \in [0,1]^m$, we will write $y \leq z$ to denote that the inequality holds component-wise, i.e., $y_g \leq z_g$ for each good $g \in [m]$. 

In this work, we particularly focus on generalized assignment constraints. 
Here, each good $g$ has an agent-specific size $s_i(g) \in \mathbb{R}_+ $ and each agent is constrained to receive goods of cumulative size at most her budget $B_i \in \mathbb{R}_+$, i.e., a bundle $x_i \in [0,1]^m$ is feasible for $i$ if and only if $\sum_{g=1}^m \  s_i(g) \  x_{i,g} \leq B_i$.  We will also consider the densities of the goods with respect to their values and sizes for the agents; formally, the density of a good $g$ for an agent $i \in [n]$ is defined as $\dens_i(g) := v_i(g)/s_i(g)$.

An allocation $y=(y_1,\ldots,y_n)$ is said to be Pareto-dominated by another allocation $z=(z_1,\ldots, z_n)$ if $v_i(y_i) \leq v_i(z_i)$ for all agents $i \in [n]$, and at least one of these inequalities is strict. In the context of constraints, we say that a feasible allocation $x=(x_1, \ldots, x_n)$ is Pareto-optimal if no other {\it feasible} allocation Pareto-dominates it. 
\medskip

\noindent
\textbf{Fairness Notions.} An allocation is said to be envy-free if no agent values the bundle of any other agent more than her own. This definition of fairness is not directly applicable in settings with constraints. Consider, e.g., the case wherein we have two agents with feasible sets A and B, respectively. The first agent cannot receive any good, while the second agent is entitled to some. Here, if one strictly clings to envy-freeness, then the only fair and feasible allocation is one in which both agents get an empty bundle, since, otherwise, the first agent would envy the second one. It is, however, unreasonable to completely deprive the second agent. Motivated by such considerations and conforming to prior works in constrained fair division, we study a more applicable fairness criterion called \emph{feasible envy-freeness}. Formally, 
\begin{definition}[Feasible envy-freeness ($\FEF$)]
\label{def:fef} 
    In an allocation $x=(x_1, \ldots, x_n)$, an agent $i \in [n]$ is said to be envy-free towards agent $j \in [n]$ if for each bundle $y \leq x_j$ with $y \in \mathcal{F}_i$, we have $v_i(x_i) \geq v_i(y)$. Further, agent $i \in [n]$ is said to be envy-free towards the \emph{charity} if each bundle $y \leq x_{\charity}$ with $y \in \mathcal{F}_i$, we have $v_i(x_i) \geq v_i(y)$. A feasible allocation is said to be $\FEF$ if each agent is envy-free towards every other agent and the charity.
\end{definition}

Our proof of existence for $\FEF$ (with divisible goods) builds on an algorithm from~\cite{BKSS23} that finds fair allocations of indivisible goods. This discrete fair division result considers, for indivisible goods, an analogous fairness notion called feasible envy-freeness up to any good ($\FEFx$). Note that in the indivisible goods setting, an allocation $(A_1,A_2, \ldots, A_n)$ is a collection of pairwise-disjoint subsets of the (indivisible) goods, where $A_i \subseteq [m]$ is the bundle assigned to agent $i$ and  $A_i \cap A_j = \emptyset$ for all $i \neq j$. The $\FEFx$ criterion is defined next.\footnote{Note that this criterion extends the well-studied $\EFx$ criterion from the unconstrained setting to the current setup, which has constraints and charity.} 



 \begin{definition}[Feasible envy-freeness up to any good ($\FEFx$)]
     In an allocation $(A_1, \ldots, A_n)$ of indivisible goods, an agent $i \in [n]$ is said to be $\FEFx$ towards agent $j$ if for every strict subset $S \subsetneq A_j$ with $S \in \mathcal{F}_i$, we have $v_i(A_i) \geq v_i(S)$. Further, agent $i$ is said to be $\FEFx$ towards the charity $A_{\charity} := [m] \setminus \left( \bigcup_{i=1}^n A_i \right)$ if for every strict subset $S \subsetneq A _{\charity}$ with $S \in \mathcal{F}_i$, we have $v_i(A_i) \geq v_i(S)$.   A feasible allocation $(A_1, \ldots, A_n)$ is said to be $\FEFx$ if each agent is $\FEFx$ towards every other agent and the charity.
  \end{definition}
  

\noindent
\textbf{Truthful Mechanisms.} For an impossibility result (in Section \ref{sec:counterexample}), we consider the following mechanism design setup. The agents (strategically) report their additive valuations for all the goods. In particular, the reports for the agents $i \in [n]$ are the values $v_i = \{v_{i,g}\}_{g \in [m]}$. The constraints are publicly known; this, in particular, avoids infeasible allocations. A mechanism $\mathcal{M}: (v_1, v_2, \dots, v_n) \mapsto [0,1]^{n \times m}$, maps reported valuation profiles to feasible allocations. 

Mechanism $\mathcal{M}$ is said to be \emph{truthful} if for  each agent $i \in [n]$ with true valuation $v_i$, and any valuation profile $v'_1, v'_2, \dots, v'_n$ we have $v_i(x_i) \geq v_i(x_i')$, where $x = \mathcal{M}(v'_1, v'_2, \ldots,v'_{i-1}, v_i, v'_{i+1}, \ldots v'_n)$ and $x' = \mathcal{M}(v'_1, v'_2, \ldots,v'_{i-1}, v_i', v'_{i+1},\dots, v'_n)$. That is,  a mechanism $\mathcal{M}$ is truthful if no agent can gain by misreporting her valuation. 

\section{Existence of FEF Allocations} 
\label{fefgeneral}

In this section, we show that there always exists an $\FEF$ allocation for a weaker class of constraints called \emph{closed} constraints \footnote{Recall that a set is closed if every convergent sequence of points in the set has a point in the set as its limit.} (where for each agent, the feasible subset is closed) and Lipschitz-continuous valuations. 

\begin{algorithm}
\caption{\ComputeFEFx{}~\cite{BKSS23} }
\begin{flushleft}
\hangindent=2em    
\hangafter=1 \textbf{Input:} Fair division instance $\langle [n], [m], \{v_i(S) \}_{i, S \subseteq [m]}, \{\mathcal{F}_i\}_i \rangle$ with indivisible goods and constraints.  

\textbf{Output:} An $\FEFx$ allocation. 
\end{flushleft}
\label{algorithm:fefx}
\begin{algorithmic}[1]
\STATE Initialize allocation $\mathcal{A}=(A_1,\ldots,A_n)=(\emptyset, \ldots,\emptyset)$ and charity $C = [m]$. 
\WHILE{the charity $C$ is envied by some agent} 
	\STATE Select a minimal envied set $T\subseteq C$ and let $k$ be the agent that envies $T$. 
	\STATE Update bundle $A_{k}\leftarrow T$ and charity $C\leftarrow [m]\setminus \left( \cup_{i=1}^n A_i \right)$.
\ENDWHILE\label{line:outer-loop-exit}
\RETURN Allocation $\mathcal{A}$.
\end{algorithmic}
\end{algorithm}

We use an algorithm from~\cite{BKSS23} to 
construct a sequence of allocations with progressively decreasing envy and argue that the limit of the sequence will be $\FEF$.
 For the sake of completeness, we reproduce the algorithm for the indivisible goods case here. It relies on the idea of iteratively swapping \emph{minimally envied subsets} from the charity (if such sets exist). A set of (indivisible) goods $T$ is called an envied set if there exists a subset $S \subseteq T$  which is feasible for some agent $k \in [n]$ who also envies it.
  \begin{definition}
       A set of goods $T$ is called a minimally envied set if $T$ is envied by some agent $k \in [n]$, while no strict subset $T' \subseteq T$ is envied by any agent $k'\in [n]$.
  \end{definition}

The proof in \cite{BKSS23} shows that for any fair division instance with indivisible goods and arbitrary constraints, Algorithm \ref{algorithm:fefx} always terminates and returns an \FEFx{} allocation.

To adapt this algorithm to the divisible goods case, we \emph{cut} each divisible good into many identical pieces. We assume that the valuation functions of all agents are $L$-Lipschitz continuous (for some finite $L$), i.e., for any allocation and any agent, the \emph{additional} value gained by adding a $\delta$-fraction of any good $g$ is at most $\delta L$. Each good $g$ is cut into $L/\varepsilon$ pieces. This ensures that each piece has value at most $\varepsilon$ for each agent. The modified algorithm is given as Algorithm \ref{algorithm:fefeps} below. We first define an approximate feasible envy-freeness property $\FEF$-$\eps$ which Algorithm \ref{algorithm:fefeps} guarantees.

\begin{definition} [Approximate feasible envy-freeness $\FEF$-$\varepsilon$]

For a fixed $\eps > 0$, an allocation $(x_1,x_2, \dots, x_n)$ is said to be $\FEF$-$\eps$ if $ \max_{i,j} \max_{S \subseteq x_j, S \in \mathcal{F}_i} (0, v_i(S) -  v_i(x_i))  \le \eps$, i.e., the envy between any two agents is upper-bounded by $\eps$.
\end{definition}

\begin{algorithm}
\caption{\ComputeFEFeps{}}
\begin{flushleft}
\hangindent=2em    
\hangafter=1
\textbf{Input:} Fair division instance $\langle [n], [m], \{v_i(S) \}_{i, S \subseteq [0,1]^m}, \{\mathcal{F}_i\}_i \rangle$ with divisible goods,  closed constraints and $L$-Lipschitz continuous valuation functions.  

\textbf{Output:} An $\FEF$-$\eps$ allocation. 
\end{flushleft}
\label{algorithm:fefeps}
\begin{algorithmic}[1]
\STATE Split each good $g \in [m]$ into $ L/\varepsilon$ identical pieces and denote by $G$ the set of all pieces.
\STATE Initialize allocation $\mathcal{A}=(A_1,\ldots,A_n)=(\emptyset, \ldots,\emptyset)$ and charity $C = G$. 
\WHILE{the charity $C$ is envied by some agent}
	\STATE Select a minimal envied set $T\subseteq C$ and let $k$ be the agent that envies $T$. 
	\STATE 
    Update bundle $A_{k}\leftarrow T$ and charity $C\leftarrow G\setminus \left( \cup_{i=1}^n A_i \right)$.
\ENDWHILE
\RETURN Allocation $\mathcal{A}$.
\end{algorithmic}
\end{algorithm}

\begin{lemma}\label{lemma:fefeps}
    For any fair-division instance, Algorithm \ref{algorithm:fefeps} outputs a $\FEF$-$\eps$ allocation.
\end{lemma}

\begin{proof}
    In Step 1, Algorithm \ref{algorithm:fefeps} splits each good into $L/\eps$ pieces and then has the exact same steps as Algorithm \ref{algorithm:fefx}, treating each piece as an individual good.
    \citet{BKSS23} show that Algorithm \ref{algorithm:fefeps} terminates with an \FEFx{} allocation. Using the fact that each piece constructed in Step 1 in Algorithm \ref{algorithm:fefeps} has value at most $\eps$, we have that envy between any two agents is at most $\eps$.
\end{proof}

We now state the main result of this section. 
\begin{theorem}
\label{theorem:generalconstraints}
    Under closed agent-specific constraints and Lipschitz-continuous valuations, there always exists an $\FEF$ allocation.
\end{theorem}

\begin{proof}

    From Lemma \ref{lemma:fefeps}, we know that for any fixed $\eps>0$, the output of Algorithm \ref{algorithm:fefeps} is $\FEF$-$\eps$.
    Consider the output of Algorithm \ref{algorithm:fefeps} for a decreasing sequence of $\varepsilon$, say $\varepsilon_1, \varepsilon_2, \dots $, such that $\varepsilon_k \in (0,1)$ and $\varepsilon_1 \geq \varepsilon_2 \geq \varepsilon_3 \geq \dots $.
    As the space of allocations $[0,1]^{n \times m}$  is a bounded and compact set, the Bolzano-Weierstrass theorem implies that, in any infinite sequence with elements in $[0,1]^{n \times m}$, there exists a convergent subsequence. Write $(x_1^*, x_2^*, \dots, x_n^*)$ to denote the limit of the sequence. Because of the \emph{closedness} property of the constraints set, we have that the limit of any sequence of feasible allocations is also a feasible allocation. Since for every $\eps > 0$, the allocation returned has envy at most $\eps$, and due to Lipschitz-continuity, the limit point $(x_1^*, x_2^*, \dots, x_n^*)$ has zero envy. This establishes the existence of $\FEF$ allocations under closed agent-specific constraints and Lipschitz-continuous valuation functions.
\end{proof}

We note that this result, as a corollary, also gives an alternative proof for the existence of $\FEF$ allocations under generalized assignment constraints. The proof, however, does not directly lead to an algorithm to compute such an allocation. 


\section{Non-Convexity of FEF Allocations}
\label{sec:non-convex}
In this section, we show that even under generalized assignment constraints, the space of $\FEF$ allocations can be non-convex. Given this, the fact that the algorithm of~\citet{BKSS23} employs linear programming as a subroutine in the computation of $\FEF$ allocations is somewhat surprising.

\begin{theorem}
\label{thm:non-convex}
    There are fair division instances with constraints, in which the space of $\FEF$ allocations is not convex.
\end{theorem}

\begin{proof}
We use an instance with two agents $1$ and $2$ and two goods $g_1$ and $g_2$. Agent $1$ has values $v_{1,g_1}=2$ for good $g_1$ and $v_{1,g_2}=1$ for good $g_2$ and size $s_{1,g_1}=s_{1,g_2}=1$ for both goods. Agent $2$ has value $v_{2,g_1}=v_{2,g_2}=1$ for both goods, and sizes $s_{2,g_1}=1$ for good $g_1$ and $s_{2,g_2}=2$ for good $g_2$. Both agents have a budget of $B_1=B_2=1$.

Consider the two allocations $x$ and $y$. In allocation $x$, agent $1$ gets good $g_2$, agent $2$ gets half of good $g_1$, and the charity gets the other half of good $g_1$. In allocation $y$, agents $1$ and $2$ get half of good $g_1$ and the charity gets good $g_2$. Both allocations are feasible as the total size of goods allocated to each agent does not exceed the budget of $1$. Indeed, agent $1$ has size $1$ for each good and gets at most one good in each allocation. Agent $2$ gets half of item $g_1$ in both allocations.

We claim that both allocations $x$ and $y$ are $\FEF$. In allocation $x$, agent $1$ has value $1$ for the good $g_2$ allocated to her. Her value for half of good $g_1$ that agent $2$ and the charity get is only $1$. Agent $2$ has value $0.5$ for half of the good $g_1$ allocated to her. She clearly does not envy the charity that gets the other half of good $g_1$. Regarding the bundle of agent $1$, only half of the good $g_2$ that is allocated to her is a feasible set for agent $2$ (as $s_{2,g_2}=2$ and $B_2 =1$), which gives her a value of $0.5$, too. 

In allocation $y$, both agents $1$ and $2$ get half of good $g_1$ and do not envy each other. Agent $1$'s value is $1$ while her value for good $g_2$ that the charity gets is $1$, too. Agent $2$'s value is $0.5$ and only half of the good $g_2$ that is allocated to the charity is feasible for her, for a value of $0.5$, too. So, no agent envies the charity.

Now, consider the allocation $z$ which is the mean of the two allocations, i.e., agent $1$ gets $0.25$ of good $g_1$ and half of good $g_2$, while agent $2$ gets half of good $g_1$. Now, observe that agent $2$ envies agent $1$. Her value for half of good $g_1$ is $0.5$. Among the goods allocated to agent $1$, $0.25$ of good $g_1$ and $0.375$ of good $g_2$ form a feasible bundle (as agent $2$ has sizes $1$ and $2$ for these goods). Agent $2$'s value for this bundle is $0.625$, i.e., higher than the value for her bundle. Hence, allocation $z$ is not $\FEF$.
\end{proof}

\section{Existence of \FEF{} and \PO{} Allocations} 
\label{sec:fefpo}
This section establishes, via a fixed-point argument, that an \FEF{} and \PO{} allocation always exists under the generalized assignment constraints and additive valuations. We utilize ideas from \cite{caragiannisppad} and \cite{hansenppad} to prove the existential guarantee and also establish the {\rm PPAD} membership of the corresponding computation problem. 

Using a fixed-point argument to prove existence involves construction of a continuous map from some convex domain set $S$ onto itself, i.e., $f:S \mapsto S$, and showing that the fixed-point of $f$ (guaranteed by Brouwer's fixed-point theorem) has the desired properties. If the function $f$ can be expressed as an arithmetic circuit using the gates $\{+,-,\max, \min, \times \zeta \}$ (where the last gate corresponds to multiplication by a rational constant), then the argument establishes PPAD-membership as well.

\citet{hansenppad} construct another gate, called linear-OPT-gate, which can be used to solve convex optimization problems satisfying certain conditions.  Suppose that the function $f$ involves solving a convex optimization problem. Then we replace the convex program by a linear-OPT gate to get a  \emph{pseudo-circuit}; under the gate’s standard preconditions, at every fixed point of the pseudo-circuit the gate’s output equals an optimal (or feasible, as appropriate) solution to that program. Additionally, they show that any pseudo-circuit using linear-OPT gates can be compiled into a pure piecewise-linear (PL) circuit with the same fixed points; since the compiled PL circuit maps a compact convex set to itself, Brouwer’s fixed-point theorem guarantees a fixed point (and thus the original pseudo-circuit has one as well). See Section $3$ and Theorem $3.1$ in \cite{hansenppad} for a detailed explanation.

In their work, a linear-OPT-gate construction is given for two types of convex optimization problems, linear programs and feasibility programs. We briefly give the description of both  problems and the pre-requisite conditions (for the construction of linear-OPT-gate). The optimization problems might be parameterized by some variables of the domain. These are referred to as gate-inputs for the problem. Also, $R$ is assumed to be some pre-defined constant in the programs below.
\medskip

\noindent\textbf{Linear Program}
\begin{align*}
\min\ c^\top x  \quad  & \text{s.t.} \quad  Ax \le b \ \  \text{ and } \ \ x \in [-R,R]^n .
\end{align*}

A linear-OPT-gate can be constructed for the above linear program provided that:
\begin{enumerate}
    \item The feasible region is non-empty and contained in a known box $[-R,R]^n$.
    \item All gate inputs appear only on the right-hand side of the constraints or in the objective: the matrix $A$ is fixed (independent of gate-inputs), while $b$ and the objective coefficients $c$ may depend on the gate-inputs.
\end{enumerate}
\medskip
\noindent\textbf{Feasibility Program}
\begin{align*}
h(y) > 0 &\implies a^Tx \le b \\
&x \in [-R,R]^n .
\end{align*}

An OPT gate can be constructed for the above program if 
\begin{enumerate}
    \item The program is feasible.
    \item the gate inputs appear only on the right-hand side of the constraints $a^Tx \le b$.
    \item In the condition constraint $h(y) > 0$, $h$ is computable using a linear arithmetic circuit and the variables are gate inputs ($x$ is not permitted).
\end{enumerate}

Given a fair division instance, we construct a linear arithmetic circuit such that the fixed-point of the circuit corresponds to an \FEF{} and \PO{} allocation proving existence and PPAD-membership of the problem.

Before stating and proving our main result, we establish some key lemmas. First, we use the following observation to characterize the set of Pareto-optimal allocations.

\begin{fact}
\label{theorem:WtSW-PO}
Let $w \in \mathbb{R}_+^n$ be any $n$-dimensional (weight) vector, with $w_i >0$ for each $i \in [n]$. Then, any feasible allocation $x=(x_1, \ldots, x_n)$ that maximizes $\langle w , v(x) \rangle := \sum_{i=1}^n w_i \cdot v_i(x_i)$
 is Pareto-optimal. 
\end{fact}


Second, we define the concept of \emph{envy graph} for a given allocation $x$ to capture the existing envy among the agents. For an allocation $x$, an envy graph is a directed graph with $n$ nodes and a directed edge from agent $i$ to $h$ if $i$ envies $h$ (according to Definition \ref{def:fef}). Note that if an allocation is $\PO$, it is implied that no agent envies the charity. Otherwise, the envied bundle from the charity can be swapped with the envious agent leading to a Pareto improvement.

\begin{lemma} 
\label{lemma:acyclic}
For any (weight) vector $w \in \mathbb{R}_+^n$, with all positive components, and any feasible allocation $x=(x_1, \ldots, x_n)$ that maximizes $\sum_{i=1}^n w_i \ v_i(x_i)$, the envy graph is acyclic. 
\end{lemma}

\begin{proof}
    This is a standard observation in the unconstrained setting. We prove that the lemma holds even in the constrained setup.
    Assume, towards a contradiction, that the envy graph for $x$ contains a cycle among the agents: $i_1 \rightarrow i_2 \rightarrow \dots \rightarrow i_c \rightarrow i_1$. Here, agent ${i_{j}}$ envies ${i_{j+1}}$, for all $j \in [c]$, and we label ${i_{c+1}} = i_1$. By definition of feasible envy, it follows that there exists a feasible envied bundle (according to agent $i_j$) in agent $i_{j+1}$'s bundle. For every agent $i_j$, denote by $y_{i_j}$ the envied bundle in agent $i_{j+1}$'s bundle. Define a new allocation $x'=(x'_1, \ldots, x'_n)$ as follows: for all agents $i \in \{i_1,i_2, \dots, i_c \}$, set $x'_i = y_i$. For all other agents, set $x'_i = x_i$. We discard the unassigned fractions of the goods to the charity. Going from allocation $x$ to $x'$, the values of the agents in the cycle strictly increase, and for all the other agents the values remain unchanged. Also, all the weights satisfy $w_i>0$. Hence, the weighted social welfare of $x'$ is strictly greater than that of $x$. This, however, contradicts the optimality of $x$ with respect to the weighted social welfare $\sum_{i=1}^n w_i \cdot v_i(x_i)$. Therefore, by way of contradiction, we obtain that the envy graph of $x$ has be acyclic. The lemma follows. 
\end{proof}

For the fixed-point argument, we set the parameter $\gamma >0$  to be sufficiently small and positive such that:

\begin{enumerate}
  \item $\gamma \le \tfrac12.$

  \item $\displaystyle 
  \gamma < \min_{\substack{i,h\in[n]\\ g\in[m]}} 
  \frac{v_i(g)}{v_h(g)}.$

  \item If the set $\{(i,h,g,g'):\ \dens_i(g')>\dens_i(g)\}$ is nonempty, then
  \[
  \gamma < \frac12
  \min_{\substack{i,h\in[n]\\ g,g'\in[m]\\ \dens_i(g')>\dens_i(g)}}
  \frac{\,v_i(g') - v_i(g) \cdot \,\dfrac{s_{i,g'}}{s_{i,g}}\,}{\,v_h(g')\,}\,;
  \]
\end{enumerate}


Note that $\gamma$ is always strictly positive. Further, the sufficiently small value of $\gamma$ gives us the following lemma.  
 
\begin{lemma} 
\label{lemma:weight}
Let $w \in \mathbb{R}_+^n$ be any vector with all positive components and $x=(x_1,\ldots,x_n)$ be any feasible allocation that maximizes $\sum_{i=1}^n w_i \cdot v_i(x_i)$. Also, let $i \in [n]$ and $h \in [n]$ be any two agents with $w_h \leq \gamma w_i$. Then, agent $i$ does not envy agent $h$ under allocation $x$.
\end{lemma}
\begin{proof}
     For the analysis, we consider the cases where $s_i(x_i) < B_i$ and $s_i(x_i) = B_i$ separately.
     
     For the first case, we show that $s_i(x_i) < B_i \implies s_i(x_h) = 0$, which implies envy-freeness. If $s_i(x_i) < B_i$ and $s_i(x_h) > 0$, the second constraint on $\gamma$ states that a small fraction of (any) good in $x_h$ can be transferred from agent $h$ to agent $i$ leading to an increase in the objective function, contradicting the assumption on $x$.
     
     If $s_i(x_i) = B_i$, we prove that, $w_h \le \gamma \cdot w_i  \implies \dens_i(g) \geq \dens_i(g')$ for all goods $g$, with $x_{i,g} > 0$, and $g'$, with $x_{h,g'} > 0$. These  density relations between the goods in $i$'s bundle and the ones in $h$'s bundle imply that $i$ does not envy $h$.  
    
    Assume, towards a contradiction, that there exist goods $g, g'$ that violate the above inequality, i.e., $\dens_i(g) < \dens_i(g')$ and $x_{i,g} > 0$ along with $x'_{h,g'}>0$. We will prove that, in such a case, we can construct a new allocation $x'$ that has a higher weighted 
    social welfare than $x$.  
    
    To construct $x'$, keep the bundles of the remaining agents (except $i$ and $h$) the same as in $x$. Write $\alpha = \min \left\{ x_{h,g'},  \frac{x_{i,g} \cdot s_{i,g}} {s_{i,g'}} \right\}$ and $\beta = \alpha \cdot \frac{s_{i,g'}}{s_{i,g}}$. Set $x'_{i,g'} = x_{i,g'} + \alpha$ and $x'_{h,g'} = x_{h,g'} - \alpha$ along with $x'_{i,g} = x_{i,g} - \beta$. 
    The allocations for items other than $g$ and $g'$ remain unchanged. 
    
    From agent $h$, a fraction of good $g'$ is transferred to agent $i$. Hence, the quantity $s_h(x_h')$ is strictly smaller than $s_h(x_h)$,  which implies that the budget constraint continues to hold for agent $h$.  For the size of agent $i$'s bundle, we have $s_i(x_i') = s_i(x_i) + \alpha \cdot s_i(g') - \beta \cdot s_i(g) = s_i(x_i)$, where the last equality follows from the definitions of $\alpha$ and $\beta$. It follows that the allocation $x'$ is feasible.

    Finally, we show that $w_i \cdot v_i(x'_i) + w_h \cdot v_h(x'_h) > w_i\cdot v_i(x_i) + w_h\cdot v_h(x_h)$. As the bundles of the rest of the agents remain unchanged, we arrive at a contradiction to the fact that $x$ maximizes $\sum_{i=1}^n w_i\cdot v_i(x_i)$.
    Note that $w_i \cdot v_i(x'_i) + w_h \cdot v_h(x'_h) = w_i \cdot v_i(x_i) + w_h \cdot v_h(x_h) + w_i \cdot \alpha \cdot v_i(g') - w_h \cdot \alpha v_h(g') - w_i \cdot \beta \cdot v_i(g)$. Hence, it suffices to prove that $w_i \cdot \alpha \cdot v_i(g') - w_h \cdot \alpha \cdot v_h(g') - w_i \cdot \beta \cdot v_i(g) > 0$.

\begin{align*}
&w_i \cdot \alpha \cdot v_i(g') - w_h \cdot \alpha \cdot v_h(g') - w_i \cdot \beta \cdot v_i(g)\\
&= \alpha \cdot w_i\cdot \left( v_i(g') - \tfrac{w_h}{w_i} \cdot v_h(g') - v_i(g)\cdot \tfrac{s_{i,g'}}{s_{i,g}} \right) \\
&\ge \alpha \cdot w_i \cdot \left( v_i(g') - \gamma \cdot v_h(g') - v_i(g) \cdot \tfrac{s_{i,g'}}{s_{i,g}} \right) \\
&> 0 .
\end{align*}
    

    The last inequality follows from the last constraint in the definition of $\gamma$. The lemma follows. 
\end{proof}


 Before proving the main result, we define a map such that its fixed-point corresponds to an \FEF{} and \PO{} allocation. Write $M = \max_i  \sum_{g \in [m]} v_i(g) $ and $W_{\eps} := \{w \ge 0; \sum_{i=1}^n w_i = 1, w_i \ge \eps\}$ . Given a fair division instance, we construct an  arithmetic linear circuit that defines a map  $F(X,\mathcal{V}, W_\eps) \;\to\; (X,\mathcal{V}, W_\eps)$, where $X = [0,1]^{n \times m}$, $\mathcal{V} = [0,M]^{n \times n}$.  Here $x$ is an allocation, $w$ is a weight vector and $v_{i,j}$ is used to denote agent $i$'s value for the maximal-valued feasible subset of agent $j$'s bundle. Also, note that all variables are bounded and hence satisfy the $[-R,R]$ condition stated above.

We give the construction of $F$ in three separate parts $P_1, P_2, P_3$ such that $F(x,v,w) = (P_1(w), P_2(x), P_3(x,P_2(x)))$.
\medskip


\noindent
{\it Part $P_1(\cdot)$:}    
    The first linear program $P_1(w)$ is defined as follows.
\[
\begin{aligned}
\max_{x}\;& \sum_{i\in[n]} w_i \cdot v_i(x_i)\\
\text{s.t.}\;& \sum_{i\in[n]} x_{i,g} \;\le\; 1 
   &\quad&\forall\,g\in[m],\\
&   x_{i,g} \ge 0 &\quad&\forall\,g\in[m], i \in [n],
\\
&   \sum_{g\in[m]} s_{i,g} \cdot x_{i,g} \;\le\; B_i
   &&\forall\,i\in[n].
\end{aligned}
\]

 $P_1$ is parameterized by vector $w \in \mathbb{R}_+^n$ (which serve as gate-inputs) and it  maps to the allocations that maximize $\sum_{i=1}^n w_i v_i(x_i)$ and, hence, are Pareto-optimal. 

Also, a linear-OPT-gate can be constructed for $P_1$ as it satisfies the conditions mentioned earlier. The domain is non-empty (setting all $x_{i,g} = 0$ satisfies the constraints) and the gate-input $w$ only figures in the objective function and not in the constraints. 
\medskip

\noindent
{\it Part $P_2(\cdot)$:}   The second linear program $P_2$ is used to compute $v_{i,h}$ for allocation $x$. We write the linear program for a fixed $i$ and $h$; the analogous $LP$ is included for every pair $(i,h) \in [n]^2$.
$P_2$ takes as input allocation $x$ and its output $P_2(x)$ is said to be the matrix $V(x)= (v_{i,h})_{i,h}$ which is used as gate-input for circuit $P_3$.
\[
\begin{aligned}
\max_{y}\;& v_i(y)\\
\text{s.t.}\;& 0 \le y_g \le x_{h,g}
   &\quad&\forall\,g\in[m],\\
&   \sum_{g\in[m]} s_{i,g}\,y_{g} \;\le\; B_i
   &&\text{budget constraint for agent } i,
   \\
&   v_{i,h} = v_i(y)
\end{aligned}
\]

Here $y \in [0,1]^m$ is a set of some auxiliary variables which corresponds to agent $i$'s maximal-valued feasible subset from agent $h$'s bundle. The constraint $v_{i,h} = v_i(y)$ ensures correctness of $v_{i,h}$. 
Also, a linear-OPT-gate can be constructed for $P_2$ as the domain of the $LP$ is non-empty (setting $y = 0$ satisfies the constraints) and $x$ (the gate-input) appears only on the right-hand side of the constraints. 
\medskip

\noindent
{\it Part $P_3(\cdot)$:}   We write $P_3$ to be a feasibility program. Set $\eps = \gamma^n / n$. We state the implication for a fixed pair $(i,h)$; the analogous constraint is imposed for all pairs $[n]^2$.
\[
\begin{aligned}
v_{i,h} - v_i(x_i) > 0&\implies w_h - \gamma \cdot w_i \le 0, \\
\sum_{i=1}^n w_i &= 1 \ \text{ and }  \ \  w_i \ge \varepsilon \ \  \forall i \in [n].
\end{aligned}
\]


 $P_3$ takes as input any weighted welfare-maximizing allocation $x$ and finds a vector $w \in \mathbb{R}^n_+$ that satisfies the following constraints: $v_{i,h} > v_i(x_i) \implies w_h - \gamma \cdot w_i \le 0$, where $v_{i,h}$ is the value of the maximally valued feasible subset of agent $h$'s bundle according agent $i$. That is, in allocation $x$, if agent $i$ envies agent $h$, then the weight $w_h$ must be sufficiently smaller than $w_i$. Further, the last set of constraints ensure that each component of $w$ is at least $\varepsilon$ and, hence, is strictly positive. 

To show that a linear-OPT-gate can be constructed for $P_3$, we first observe that the gate-inputs $v_{i,j}$ and $x$ do not figure in the linear enforced constraints and only appear on the left-hand side of the conditional constraint. This is in line with the conditions stated above. Finally, we need to show that $P_3$ is feasible. We make use of the following lemma from \cite[Lemma 3.5]{caragiannisppad}.

\begin{lemma}
\label{lemma:weightvector}
    (\citet{caragiannisppad})
    Suppose $x$ is an allocation such that the envy graph of $x$ is acyclic. Then, the program $P_3$ is feasible.
\end{lemma}

Lemma \ref{lemma:acyclic} shows that for all weighted social welfare maximizing allocations, the envy graph is acyclic. Applying Lemma \ref{lemma:weightvector}, we get that for any weighted social welfare maximizing allocation $x$, a valid $w$ vector can be found, i.e., $P_3$ is feasible. This concludes the construction of the function $F$. Using this function, we next prove the main result of this section.



\begin{theorem}
\label{thm:fefpo}
    For any given fair division instance with divisible goods and generalized assignment constraints, there always exists an $\FEF$ and $\PO$ allocation. Moreover, the problem of finding such an allocation is in PPAD.
\end{theorem}

\begin{proof}



As noted earlier, the circuit we construct using linear-OPT-gates compiles to a piecewise-linear map which preserves fixed points. Further, by Brouwer's fixed-point theorem, we have that there exists at least one fixed-point  for the map $F$ (for the piecewise-linear map and, hence for $F$ due to preservation of fixed points). Let one such fixed-point be $(x^*,v^*,w^*)$. We show that $x^*$ is an \FEF{} and \PO{} allocation. 

First, observe that $x^*$ maximizes $\sum_{i\in[n]} w_i^*\cdot v_i(x_i)$ and hence is Pareto-optimal. Next, we show that $x^*$ is \FEF{}. Assume towards a contradiction, that $x^*$ is not $\FEF$, i.e., an agent $i$ (feasibly) envies $h$. As $v_{i,h} > v_i(x_i)$, by the constraints in $P_3$ we have, $w^*_h \leq \gamma w^*_i$. But, Lemma \ref{lemma:weight} implies the $i$ cannot envy $h$ (as $w^*_h \leq \gamma \cdot w^*_i$). This leads to a contradiction, and the theorem follows.
\end{proof}




\section{Are Fairness, Efficiency, and Truthfulness Compatible?}
\label{sec:counterexample}

We now prove that truthful mechanisms that compute \FEF{} and \PO{} allocations do not exist.\footnote{We remark that the proof of Theorem~\ref{thm:impossibility} can be obtained from~\cite{gp22} and 
\cite{Momi2017}. We present an alternative proof here.} Actually, our proof uses a fair division instance without constraints. Then, \FEF{} is identical to classical envy-freeness, and Pareto-optimality implies that all the goods are allocated to the agents and not to the charity.

\begin{theorem}\label{thm:impossibility}
    There does not exist a truthful mechanism that always outputs envy-free and Pareto-optimal allocations. 
\end{theorem}

Our proof uses the following fair division instance with two agents and two goods. The valuations of the first and second agent are given in the left and right table, respectively.


\begin{table}[h]
\centering
\begin{tabular}{|c|c|}
\hline
 Good & Value \\
\hline
 1 & $\alpha$ \\
\hline 
 2 &  $1-\alpha$ \\
\hline
\end{tabular}
\hspace{20pt}
\begin{tabular}{|c|c|}
\hline
 Good & Value \\
\hline
 1 & $\beta$  \\
\hline 
 2 &  $1-\beta$\\
\hline
\end{tabular}
\caption{The valuations in the proof of Theorem~\ref{thm:impossibility}.}
\label{table:size-val}
\end{table}

\noindent The instance is parameterized by two variables $\alpha, \beta \in (1/2,1)$ and $\alpha > \beta$. Note that the first good is always more valuable than the second good for both agents. 

\begin{claim}
\label{lemma:counterexample_PO}
    For every $\PO$ allocation $x$, either $x_{1,1}=1$ or $x_{1,2}=0$.
\end{claim}

\begin{proof}
    Towards a contradiction, assume none of the two conditions stated are true, i.e., $x_{1,1} < 1$ and $x_{1,2}>0$. We define allocation $x'$ as follows. For sufficiently small $\varepsilon$ and $\varepsilon'$ such that $\beta/(1-\beta) < \varepsilon'/\varepsilon < \alpha/(1-\alpha)$,  set $x'_{1,1} = x_{1,1} + \varepsilon$, $x'_{1,2} = x_{1,2} - \varepsilon'$, $x'_{2,1} = x_{2,1} - \varepsilon$ and $x'_{2,2} = x_{2,2} + \varepsilon'$. It can be verified that the $x'$ Pareto-dominates $x$, hence establishing a contradiction to the Pareto-optimality of $x$.
\end{proof}

\begin{claim}
\label{lemma:counterexample_EF}
    For every envy-free and Pareto-optimal allocation $x$, we have $x_{1,2} = 0$
and $\frac{1}{2\alpha} \leq x_{1,1} \leq \frac{1}{2\beta}$.
\end{claim}

\begin{proof}
    We note that for any envy-free allocation, it holds $x_{1,1} < 1$, otherwise agent $2$ would envy agent $1$. Hence, from Claim \ref{lemma:counterexample_PO}, we have $x_{1,2} = 0$. Next, we write the condition for envy-freeness for both the agents. For agent $1$, we have $\alpha \cdot x_{1,1}\, \ge (1 - \alpha) \cdot x_{1,1}\, + 1 - \alpha$, which implies $x_{1,1} \ge \frac{1}{2\alpha}$. For agent $2$, we have $\beta \cdot x_{1,1}\,  
  \le  \beta \cdot x_{2,1}+ (1 - \beta) \cdot x_{2,2} \leq (1 - x_{1,1})\cdot \beta + 1 - \beta$, which implies $x_{1,1} \le \frac{1}{2\beta}$.
\end{proof}

 We are now ready to complete the proof of Theorem \ref{thm:impossibility}.
\medskip

{\sc Proof of Theorem~\ref{thm:impossibility}.}
    Towards a contradiction, assume that some mechanism $\mathcal{M}$ satisfies all the three properties. Consider the fair division instance defined above and let the  output of  $\mathcal{M}$ for the instance be $x$. We show that if $\mathcal{M}$ returns an envy-free and Pareto-optimal allocation, then it is not truthful.

From Claim \ref{lemma:counterexample_EF}, we know that $\mathcal{M}$ returns $x$ such that $1/(2\alpha) \leq x_{1,1} \leq 1/(2\beta)$. 
Note that agent $1$ (resp. $2$) prefers higher (resp. lower) fraction $x_{1,1}$.
If $x_{1,1} = 1/(2\alpha)$, then agent $1$ has an incentive to report her value to be $\alpha' = (\alpha+\beta)/2$. 
In this case, $\alpha' < \alpha$ and $1/(2\alpha') > 1/(2\alpha)$, which renders the output $x_{1,1} = 1/(2\alpha)$ invalid due to Claim \ref{lemma:counterexample_EF}. Hence, here $\mathcal{M}$ has to output a higher value of $x_{1,1}$.

Similarly, if $x_{1,1} > 1/(2\alpha)$, then agent $2$ has an incentive to misreport $\beta$. Specifically, if agent $2$ misreports her values as $\beta'$ such that $x_{1,1} >  1/(2\beta')$ (while still maintaining $\beta'< \alpha$), then $\mathcal{M}$ has to output a lower value of $x_{1,1}$ to ensure $x_{1,1} \leq 1/(2\beta')$.

To conclude, no matter the value of $x_{1,1}$ in the output of $\mathcal{M}$, at least one of the agents has an incentive to misreport her valuation. It follows that no mechanism which always outputs envy-free and Pareto-optimal allocations can be truthful. \qed

\subsection{Truthfulness}
\label{section:truthful-in-pairs}
In this section, we complement the impossibility result by showing compatibility of truthfulness individually with Pareto-optimality and feasible envy-freeness.
\medskip

\noindent
\textbf{Truthfulness and Envy-Freeness.}
\label{section:truthful-and-fef}
We devise a truthful mechanism $\mathcal{M}$ which, given a fair division instance under generalized assignment constraints with two agents, always outputs an \FEF{} allocation. $\mathcal{M}$ creates two disjoint subsets of goods, one for each agent, $S_1$ and $S_2$. Each good $g$ is split into identical halves, say $g^{(1)}$ and $g^{(2)}$ and assigned to $S_1$ and $S_2$ respectively.  Then for each agent $i \in \{1,2\}$, $\mathcal{M}$ sets $x_i$ to be the maximum-valued feasible subset of $S_i$ for agent $i$.

\begin{algorithm}[h]
\caption{Truthfulness +\PO{}}

\begin{flushleft}
\hangindent=2em    
\hangafter=1 
\textbf{Input:} Fair division instance $\langle [n], [m], \{v_i\}_{i}, \{\mathcal{F}_i\}_i \rangle$ 

\textbf{Output:} A $\PO$ allocation. 
\end{flushleft}
\label{algorithm:truthPO}
\begin{algorithmic}[1]
\STATE $\mathcal{P} \gets \varnothing$ \Comment{Constraints to be accumulated}
\FOR{$i=1$ to $n$} \STATE Let $q_i$ be the optimal value of stage-$i$ LP with variables $\{x_{j,g}\}_{j\le i,\,g\in[m]}$: \[ \quad\quad\quad\begin{aligned} \text{maximize}\quad & v_i(x_i)\\ \text{subject to}\quad & x_{j,g} \ge 0 \quad \forall j\le i,\ \forall g\in[m]\\ & \sum_{j=1}^{i} x_{j,g} \le 1 \quad \forall g\in[m]\\ & x_j \in \mathcal{F}_j \quad \forall j\le i\\ & \mathcal{P}  
\Comment{constraints accumulated so far} \end{aligned} \]
\STATE $\mathcal{P} \gets \mathcal{P} \cup \{\, v_i(x_i)=q_i \,\}$ \Comment{Set agent $i$'s value} \ENDFOR
\STATE Solve the following final (feasibility) LP to find values for the variables $\{x_{j,g}\}_{j\le n,\,g\in[m]}$ that satisfy: \[ \begin{aligned} 
& x_{j,g} \ge 0 \quad \forall j\in[n],\ \forall g\in[m]\\ & \sum_{j=1}^{n} x_{j,g} \le 1 \quad \forall g\in[m]\\ & x_j \in \mathcal{F}_j \quad \forall j\in[n]\\ & \mathcal{P} \end{aligned}\] \STATE \textbf{Return} $\mathcal{A}=(x_1,\dots,x_n)$ \end{algorithmic} \end{algorithm}

\begin{theorem}
\label{theorem:fef-truthful-two-agents}
    For a fair division instance with $n=2$ agents and generalized assignment constraints, $\mathcal{M}$ is truthful and always outputs an \FEF{} allocation.  
\end{theorem}

\begin{proof}

First, we show that $\mathcal{M}$ is truthful. The mechanism computes $x_i \in \arg \max \{v_i(x): x \in S_i  \text{ and $x$ is feasible for $i$} \}$, where the feasible set is independent of reported valuations. Thus, reporting truthfully maximizes the utility of each agent.

To complete the proof, we establish the \FEF{} property of the allocation by showing that agent $1$ does not envy agent $2$ and the charity. By symmetry, the argument also implies that agent $2$ does not envy anyone either. 

The envy-freeness from agent $1$ to agent $2$ follows from the fact that $x_2 \subseteq S_2 = S_1$ and the definition of $x_1$. Towards a contradiction, assume that agent $1$ envies some feasible subset $x'_2 \subseteq x_2 \subseteq S_2$. Then, $v_1(x'_2) > v_1(x_1)$ contradicting the definition of $x_1$.

 We show that agent $1$ does not envy the set $(S_1 \backslash x_1) \cup S_2$, which is a superset of the set $C$ consisting of goods allocated to charity, i.e., $C \subseteq (S_1 \backslash x_1) \cup S_2$. Observe that $x_1$ is constructed by taking the most dense goods from $S_1$ until $s_1(x_1) = B_1$ or $x_1 = S_1$.  If $x_1 = S_1$, then envy-freeness towards charity follows since $C \subseteq S_2$.
 Otherwise, it holds that all goods in $S_1 \backslash x_1$ are at most as dense as the goods in $x_1$. Hence, the maximum-valued feasible subset of goods from  $(S_1 \backslash x_1) \cup S_2$ can be assumed to be picked just from $S_2$ which ensures that the value is at most $v_1(x_1)$ by definition.
\end{proof}

\medskip
\noindent
\textbf{Truthfulness and Pareto-Optimality}
\label{section:truthful-PO}
Next, we show that truthfulness is compatible with Pareto-optimality. 



\begin{theorem}
    Algorithm \ref{algorithm:truthPO} is a truthful mechanism which always outputs a \PO{} allocation.
\end{theorem}

\begin{proof}
    First, we establish that Algorithm \ref{algorithm:truthPO} always returns a \PO{} allocation.
    Towards a contradiction, assume that, there exists an instance for which Algorithm \ref{algorithm:truthPO} outputs allocation $x$ and there exists an allocation $x'$ such that for all agents $j \in [n]$, $v_j(x') \ge v_j(x)$ with at least one inequality being strict. Let $h$ be the least-indexed agent such that the inequality is strict. Then, we have that for all $j \in [h-1]$, $v_j(x'_j) = q_j$ and $v_h(x'_h) > q_h$. However this directly contradicts the maximality of $q_h$ for the stage-$h$ linear program. Specifically, setting $x_j = x'_j$ for all $j \in [h]$ achieves a higher objective value as $v_h(x'_h) > q_h$.

    Next, we show that the mechanism is truthful. Observe that the final value obtained by agent $j$ is equal to $q_j$, which is the optimal value of the stage-$j$ linear program. As the reported valuation vector appears only in the objective function and not in the constraints, it is straightforward to see that maximizing the objective $v_j(x_j)$ leads to maximum value for agent $j$.
\end{proof}

\section{Conclusion}
This work advances our understanding of fair, efficient, and truthful allocations of divisible goods under generalized assignment constraints. We resolve an open problem posed in \cite{barmanfinding} by establishing that an $\FEF$ and $\PO$ allocation always exists under this setting. We also prove that truthfulness, fairness, and efficiency are incompatible, while these three properties when considered in pairs admit positive results. Extending Theorem \ref{theorem:fef-truthful-two-agents} by developing a truthful mechanism that finds fair (but not necessarily Pareto-optimal) allocations among three or more agents, or showing the nonexistence of such a mechanism, is an interesting direction for future work. Also, exploring whether truthfulness is compatible with approximations of Pareto-optimality and $\FEF$ is another direction that deserves investigation.

\section*{Acknowledgements}
Ioannis Caragiannis and Sudarshan Shyam were partially supported by the Independent Research Fund Denmark (DFF) under grant 2032-00185B. Siddharth Barman acknowledges the support of the Walmart Center for Tech Excellence (CSR WMGT-23-0001) and an Ittiam CSR Grant (OD/OTHR-24-0032).

\bibliographystyle{ACM-Reference-Format}
\bibliography{sample}

\end{document}